\documentclass[11pt]{article}
\usepackage{amsmath}
\usepackage{amssymb}
\usepackage{amsthm}
\usepackage{mathdots}
\usepackage{tikz}
\usepackage{algorithm}
\usepackage{algorithmic}
\usetikzlibrary{automata}
\usetikzlibrary{positioning}
\usetikzlibrary{shapes.misc}
\usepackage[normalem]{ulem}
\usepackage{subcaption}
\newcommand{\defeq}{\overset{\rm def}{=}}

\newcommand{\ceil}[1]{\lceil{#1}\rceil}
\newcommand{\flip}{\overline}

\newcommand{\E}{\mathrm{E}}

\newcommand{\tent}{f}

\newcommand{\bitseq}{\mathbf{b}}
\newcommand{\bitseqc}{\mathbf{c}}

\newcommand{\enc}{\gamma}
\newcommand{\typei}{I}
\newcommand{\typej}{J}
\newcommand{\typefn}{T}

\newcommand{\typeset}{\mathcal{T}}

\newcommand{\requiredstgsize}{K}

\newcommand{\poly}{\mathrm{poly}}
\newcommand{\Order}{\mathrm{O}}
\newcommand{\order}{\mathrm{o}}
\newcommand{\lev}{L}

\newtheorem{theorem}{Theorem}[section]
\newtheorem{lemma}[theorem]{Lemma}

\newtheorem{proposition}[theorem]{Proposition}

\newtheorem{observation}[theorem]{Observation}
\newtheorem{problem}{Problem}

\title{A Smoothed Analysis of the Space Complexity\\of Computing a Chaotic Sequence}
\author{
 Naoaki Okada\footnote{Graduate School of Information Science and Electrical Engineering, Kyushu University} \and 
 Shuji Kijima\footnote{Faculty of Data Science, Shiga University}
}

\begin{document}
\maketitle

\begin{abstract}
 This work is motivated by a question 
  whether it is  possible to calculate a chaotic sequence efficiently, 
  e.g., is it possible to get the $n$-th bit of a bit sequence generated by a chaotic map, 
   such as $\beta$-expansion, tent map and logistic map in $\order(n)$ time/space?
 This paper gives an affirmative answer to the question about the space complexity of a tent map. 
 We show that 
  the decision problem of whether a given bit sequence is a {\em valid} tent code  
  is solved in $\Order(\log^2 n)$ space 
  in a sense of the smoothed complexity. 
\end{abstract}

\section{Brief Introduction}
A {\em tent map} $\tent_{\mu} \colon [0, 1] \to [0, 1]$ (or simply $\tent$) is given by 
\begin{align}\label{eq:tentmap}
  \tent(x) =
  \begin{cases}
    \mu x &: x \leq \frac{1}{2}, \\
    \mu (1 - x) &: x \ge \frac{1}{2}
  \end{cases}
\end{align}
  where this paper is concerned with the case of $1 < \mu < 2$. 
 As Figure~\ref{fig:tentmap} shows, 
   it is a simple piecewise-linear map looking like a tent.  
 Let $x_n = \tent(x_{n-1}) = \tent^n(x) $ recursively for $n=1,2,\ldots$, where $x_0=x$ for convenience. 
 Clearly, $x_0,x_1,x_2,\ldots$ is a deterministic sequence.
 Nevertheless, the deterministic sequence shows a complex behavior, as if ``random,'' when $\mu>1$. 
 It is said {\em chaotic}~\cite{Lorenz93}. 
 For instance, 
  $\tent^n(x)$ becomes quite different from $\tent^n(x')$ for $x \neq x'$ as $n$ increasing, 
    even if $|x-x'|$ is very small, and 
   it is one of the most significant characters of 
    a chaotic sequence   
    known as the  {\em sensitivity to initial conditions} --- 
    a chaotic sequence is said ``unpredictable'' despite a deterministic process~\cite{Lorenz63,SMYO83,YMS83,CH94}. 

 From the viewpoint of theoretical computer science, 
  computing chaotic sequences 
     seems to contain (at least) two computational issues: 
     numerical issues and 
     combinatorial issues including computational complexity. 
 This paper is concerned with the computational complexity of a simple problem: 
 Given $\mu$, $x$ and $n$, 
  decide whether $\tent^n(x) < 1/2$. 
 Its time complexity might be one of 
  the most interesting questions; 
   e.g., is it possible to ``predict'' whether $\tent^n(x) < 1/2$ in time polynomial in $\log n$? 
 Unfortunately, we in this paper cannot answer the question\footnote{
   We think that the problem might be NP-hard 
    using the arguments on the complexity of algebra and number theory in~\cite{GJ79}, 
    but  we could not find the fact. 
  }. 
 Instead, 
   this paper 
   is concerned with the space complexity of the problem.  

\begin{figure}[tbp]
  \begin{tabular}{cc}
    \begin{minipage}[t]{.5\hsize}
      \centering
      \includegraphics[width=1\linewidth]{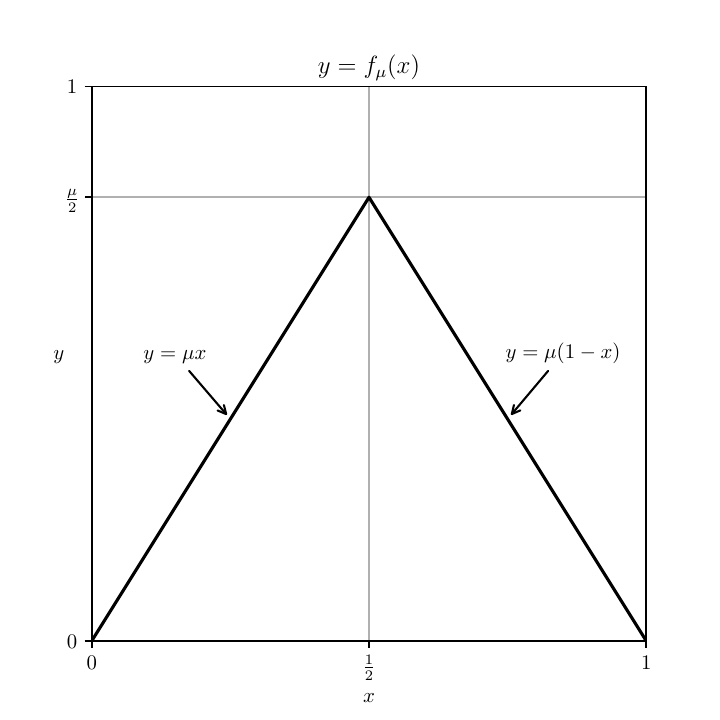}
      \subcaption{Tent map $\tent(x)$ for $\mu=1.6$. }
      \label{fig:tentmap}
    \end{minipage} &
    \begin{minipage}[t]{.5\hsize}
      \centering
      \includegraphics[width=1\linewidth]{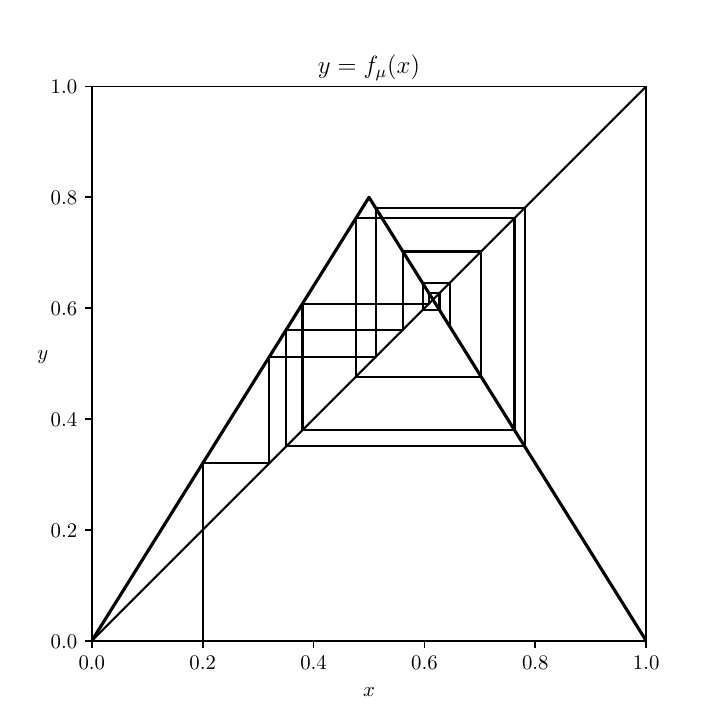}
      \subcaption{Cobweb of $x_0, x_1, \ldots, x_{15}$ for $x_0=0.2$.\\Is $x_{1000} < \frac{1}{2}$?}
      \label{fig:expansion}
    \end{minipage}
  \end{tabular}
  \caption{A tent map $\tent(x)$ and its cobweb.}
  \label{fig:compress0}
\end{figure}

\subsection{Background}
\paragraph{Chaos.}
%
%
%
%
%

 Chaotic sequences show many interesting figures 
  such as 
    cobweb, 
    strange attractor, 
    bifurcation, etc.~\cite{Lorenz63,LY75,May76,Lorenz93,Kohda08}. 
 The chaos theory has been intensively developed 
   in several context such as 
     electrical engineering, information theory, statistical physics, neuroscience and computer science,  
  with many applications, such as weather forecasting, climate change, diastrophism and disaster resilience since the 1960s. 
 For instance, 
   a cellular automaton, including the life game, is 
    a classical topic in computer science, and 
   it is closely related to the ``edge of chaos.'' 
 For another instance, 
  the ``sensitivity to initial conditions'' are often regarded as unpredictability, 
   and chaotic sequences are used in 
    pseudo random number generator,  
    cryptography, 
    or heuristics for NP-hard problems including chaotic genetic algorithms. 
 
 From the viewpoint of theoretical computer science, 
  the numerical issues of computing 
   chaotic sequences have been intensively investigated 
  in statistical physics, information theory and probability. 
 In contrast, 
   the computational complexity of computing a chaotic sequence 
    seems not well developed. 
 It may be a simple reason that 
  it looks unlike a decision problem. 
 
%


\paragraph{Tent map: 1-D, piecewise-linear and chaotic.}
 Interestingly, 
   very simple maps show chaotic behavior. 
 One of the most simplest maps are piece-wise linear maps, 
  including the tent map and the $\beta$-expansion (a.k.a. Bernoulli shift) which are 1-D maps and 
    the baker's map which is a 2-D map 
  \cite{Lorenz63,Renyi57,Parry60,Parry64,SMYO83,YMS83,CH94,Hopf37,baker00}. 

 The tent map, as well as the $\beta$-expansion, is known to be 
  topologically conjugate to the {\em logistic map} 
   which is a quadratic map cerebrated as a chaotic map. 
 Chaotic behavior of the tent map, in terms of power spectra, band structure, critical behavior, are analyzed in e.g.,~\cite{Lorenz63,SMYO83,YMS83,CH94}. 
 The tent map is also used for pseudo random generator or encryption e.g., \cite{addabbo2006,Access20,LLQL17}. 
 It is also used for  meta-heuristics for NP-hard problems~\cite{DNA09,NoC09}.
%
%

\paragraph{Smoothed analysis. }
 Linear programming is an optimization problem on a (piecewise) linear system, for a linear objective function. 
 Dantzig in 1947 gave an ``efficient'' algorithm, known as the simplex method, for linear programming. 
 Khachiyan~\cite{Khachiyan} gave the ellipsoid method, and 
  proved that the linear programming is in P (cf.~\cite{KV}) in 1979. 
 Karmarkar~\cite{Karmarkar} in 1984 gave another polynomial time  algorithm, interior point method. 
  
 The smoothed analysis is introduced by Spielman and Teng~\cite{spielman2004}, 
  to prove that the simplex algorithms for linear programmings run in ``polynomial time,'' beyond the average case analysis. 
 There are several recent progress on the smoothed analysis of algorithms~\cite{CGGYZ20,BBHR20,HRS22,GKM22}.

\subsection{Contribution}
\paragraph{This work.}
 This paper is concerned with a problem 
   related to deciding whether $\tent^n(x) < 1/2$ for $x \in [0,1)$ for the $n$-th iterated tent map $\tent^n$. 
 More precisely, we will define the {\em tent language} ${\cal L}_n \subseteq \{0,1\}^n$ 
   consisting of {\em tent codes} of  $x \in [0,1)$ in Section~\ref{sec:tent-map}, 
  and we are concerned with the {\em correct} recognition of ${\cal L}_n$. 
 The main target of the paper is the {\em space complexity} of the following simple problem; 
    given a bit sequence $\bitseq \in \{0,1\}^n$ and $x \in [0,1)$, 
   decide whether $\bitseq$ is a tent code of $x$. 
 One may think that 
    it is a problem just to compute $\tent^i(x)$ ($i=1,\ldots,n$), and 
    there is nothing more than the precision issue, even in the sense of computational complexity. 
 However, we will show in Section~\ref{sec:issue} that 
  a standard calculation attended by {\em rounding-off} easily allows {\em impossible} tent code. 

 By a standard argument on the numerical error, cf.\ \cite{KV}, 
   $\Order(n)$ space is enough to get it. 
 At the same time, 
   it seems hopeless to solve the target problem {\em exactly} in $\order(n)$ space, 
   due to the ``sensitivity to initial conditions'' of a chaotic sequence.  
 Then, this paper is concerned with a decision problem of 
   whether an input $\bitseq \in \{0,1\}^n$ is a tent code of $\epsilon$-perturbed $x$, and 
   proves that it is {\em correctly} recognized in $\Order(\log^2 n)$ space (Theorem~\ref{thm:smooth-recog}).

\paragraph{Related works.}
 The analysis technique of the paper basically follows~\cite{OK23}, 
   which showed that the recognition of ${\cal L}_n$ is in $\Order(\log^2 n)$ space {\em in average}. 
 The technique is also very similar to or essentially the same as 
   {\em Markov extension}  developed in the context of symbolic dynamics. 
 In 1979, Hofbauer~\cite{Hofbauer79} gave a representation of the kneading invariants for unimodal maps, 
  which is known as the Markov extension and/or Hofbauer tower, and then 
   discussed {\em topological entropy}. 
 Hofbauer and Keller extensively developed the arguments in 1980s, see e.g., \cite{deMelo-vanStrien,Bruin95}. 
 We do not think the algorithms of the paper are trivial, 
   but they are composed of the combination of the above nontrivial argument, and 
   some classical techniques of designing space efficient algorithms.  

 As we stated above, the computational complexity of computing a chaotic sequence seems not well developed. 
 Perl showed some NP-complete systems, e.g., knapsack, shows chaotic behavior \cite{Perl87}. 
 On the other hand, 
  it seems not known whether every chaotic sequence is hard to compute in the sense of NP-hard; 
  particularly we are not sure if the problem $\tent^n(x) < 1/2$ is NP-hard for a tent map $\tent$. 
 Recently, chaotic dynamics are used for solving NP-hard problems e.g., SAT \cite{ET11}.

\subsection{Organization}
 In Section~\ref{sec:tent-map}, 
   we will define the tent code, 
   describe the issue of a standard calculation  rounding-off, and 
   show the precise results of the paper. 
 Section~\ref{sec:OK23} imports some basic technologies from \cite{OK23}. 
 Section~\ref{sec:calc} gives a simple algorithm for a valid {\em calculation}, 
   as a preliminary step. 
 Section~\ref{sec:smooth-recog} gives a smoothed analysis for the {\em decision} problem.

\section{Issues and Results}\label{sec:tent-map}
\subsection{Tent code}\label{sec:code}
 We define a tent-encoding function $\enc^{n}_{\mu}\colon [0,1) \to \{0,1\}^n$ (or simply $\enc^n$) as follows. 
For convenience, 
 let $x_i=\tent^i(x)$ for $i=1,2,\ldots$ as given $x \in [0, 1)$, 
 where $\tent^i$ denote the $i$-times iterated tent map 
  formally given by $\tent^i(x)=\tent(\tent^{i-1}(x))$ recursively. 
Then, the {\em tent code} $\enc^{n}(x)=b_1 \cdots b_n$ for $x \in [0, 1)$ is a bit-sequence, 
  where 
\begin{align}
  b_{1} &=
  \begin{cases}
    0 &: x < \frac{1}{2}, \\
    1 &: x \ge \frac{1}{2},
  \end{cases} 
  \label{def:encode0} 
\end{align}
 and $b_i$ ($i=2,3,\ldots,n$)  is recursively given by 
\begin{align}
  b_{i+1} &=
  \begin{cases}
    b_{i} &: x_i < \frac{1}{2}, \\
    \flip{b_{i}} &: x_i > \frac{1}{2}, \\
    1 &: x_i = \frac{1}{2},
  \end{cases}
\label{def:encode1}
\end{align}
where $\flip{b}$ denotes bit inversion of $b$,
i.e., $\overline{0}=1$ and $\overline{1}=0$.
We remark that the definition \eqref{def:encode1} is rephrased by
\begin{align}
  b_{i+1} &=
  \begin{cases}
    0 &: [b_i = 0] \wedge \left[x_i < \frac{1}{2}\right], \\
    1 &: [b_i = 0] \wedge \left[x_i \geq \frac{1}{2}\right], \\
    1 &: [b_i = 1] \wedge \left[x_i \leq \frac{1}{2}\right], \\
    0 &: [b_i = 1] \wedge \left[x_i > \frac{1}{2}\right].
  \end{cases}
\label{def:encode2}
\end{align}

\begin{proposition}\label{prop:tent-expansion}
 Suppose $\enc_{\mu}^{\infty}(x) = b_1b_2\cdots$ for $x \in [0,1)$. 
 Then, $(\mu - 1) \sum_{i=1}^{\infty} b_{i} \mu^{-i} = x$. 
\end{proposition}

 See \cite{OK23} for a proof. 
 The proofs are not difficult but lengthy.  
 Thanks to this a little bit artificial definition \eqref{def:encode1},  
   we obtain the following two more facts. 
\begin{proposition}\label{prop:order}
  For any $x, x^{\prime} \in [0, 1)$,
  \begin{align*}
    x \leq x^{\prime} &\Rightarrow \enc_{\mu}^n(x) \preceq \enc_{\mu}^n(x^{\prime})  
  \end{align*}
  hold where $\preceq$ denotes the {\em lexicographic order}, 
   that is $b_{i_*}=0$ and $b'_{i_*}=1$ at $i_* = \min\{j \in \{1,2,\ldots \} \mid b_j \neq b'_j\}$
   for $\enc^n(x) =b_1b_2\cdots b_n$ and  $\enc^n(x^{\prime})=b'_1b'_2\cdots b'_n$
   unless  $\enc^n(x) = \enc^n(x')$. 
\end{proposition}

\begin{proposition}\label{prop:-heikai}
The $n$-th iterated tent code is right continuous, i.e., $\enc_{\mu}^n(x) = \enc_{\mu}^n(x+0)$. 
\end{proposition}

 These two facts make the arguments simple.  
 The following technical lemmas are useful 
  to prove Propositions~\ref{prop:order} and \ref{prop:-heikai}, 
   as well as the arguments in Sections~\ref{sec:calc} and \ref{sec:smooth-recog}. 
\begin{lemma}\label{lem:encode1}
  Suppose $x,x' \in [0,1)$ satisfy $x < x'$. 
 If $\enc^n(x) = \enc^n(x')$ then 
  \begin{align}
  \begin{cases}
    x_n < x^{\prime}_n  & \mbox{if $b_n = b^{\prime}_n = 0$, }\\ 
    x_n > x^{\prime}_n & \mbox{if $b_n = b^{\prime}_n = 1$}
 \end{cases}
  \label{eq:tentexpansion_lexicography_1}
  \end{align}
  holds. 
\end{lemma}

\begin{lemma}\label{lem:encode3}
 If $\enc^n(x) = \enc^n(x')$ for  $x,x' \in[0,1)$ then $|\tent^n(x) - \tent^n(x')| = \mu^n|x-x'|$. 
\end{lemma}

 Let ${\cal L}_{n,\mu}$ (or simply ${\cal L}_n$) denote the set of all $n$-bits tent codes, 
 i.e., 
\begin{align}
{\cal L}_n = \left\{ \enc_{\mu}^{n}(x) \in \{0,1\}^n \ \middle|\ x \in [0,1) \right\}
\label{def:tent-lang}
\end{align} 
 and we call ${\cal L}_n$ {\em tent language} (by $\mu \in (1,2)$). 
 Note that ${\cal L}_n \subsetneq \{0,1\}^n$ for $\mu \in (1,2)$. 
 We say $\bitseq_n \in \{0,1\}^n$ is a {\em valid} tent code if $\bitseq_n \in {\cal L}_n$.

\subsection{What is the issue?}\label{sec:issue}
 A natural problem for the tent code could be {\em calculation}:  
  given $x \in (0,1]$ and $n \in \mathbb{Z}_{>0}$, find $\enc^n(x)$. 
By a standard argument (see e.g., \cite{KV}), 
 it requires $\Theta(n \log \mu)$ working space to compute $\tent^n(x)$, in general.  
Thus, it is natural in practice to employ 
 {\em rounding-off}, like Algorithm~\ref{alg:calc-wrong}.  

\begin{algorithm}[t]
    \caption{Rounding-off could output an {\em invalid} code}
    \label{alg:calc-wrong}
    \begin{algorithmic}[1]
    \REQUIRE $x \in [0,1]$ 
    \ENSURE a bit sequence $b_1\cdots b_n$ \COMMENT{$ b_1\cdots b_n \not\in {\cal L}_n$ in bad cases}
    \STATE {\bf set} {\rm int} $\kappa$ large constant 
    \STATE {\rm rational} $z \leftarrow \langle x \rangle_{\kappa}$ \COMMENT{Round off by $\kappa$ bits (or digits)}
    \STATE {\rm bit} $b \leftarrow 0$
    \FOR{$i=1$ to $n$}
    \IF{$b=0$}
    \STATE  {\bf if} $z<\frac{1}{2}$ {\bf then} $b \leftarrow 0$, {\bf else} $b \leftarrow 1$\COMMENT{recall \eqref{def:encode2}}
    \ELSE
    \STATE  {\bf if} $z > \frac{1}{2}$ {\bf then} $b \leftarrow 0$, {\bf else} $b \leftarrow 1$\COMMENT{recall \eqref{def:encode2}}
    \ENDIF
    \RETURN $b$ \COMMENT{as $b_i$}
    \STATE $z \leftarrow \langle \tent(z) \rangle_{\kappa}$
    \ENDFOR
    \end{algorithmic}
\end{algorithm}
 Due to the sensitivity to initial condition of a chaotic sequence, 
   we cannot expect that Algorithm~\ref{alg:calc-wrong} to output $b_1\cdots b_n = \enc^n(x)$ exactly, 
   but we hope that it would output some {\em approximation}. 
 It could be a natural question 
  whether  the output $b_1\cdots b_n \in {\cal L}_n $. 
 The following proposition means that 
   Algorithm~\ref{alg:calc-wrong} could output an impossible tent code. 


\begin{proposition}\label{prop:round}
Algorithm~\ref{alg:calc-wrong} could output $b_1\cdots b_n \not \in {\cal L}_n$. 
\end{proposition}
\begin{proof}
Let $\mu = 1.62$, 
  which is slightly greater than the golden ratio $\frac{1+\sqrt{5}}{2} \simeq 1.61803\ldots$, 
  where  the golden ratio is a solution of $\mu^2 - \mu - 1 = 0$. 
By our calculation, $\enc_{\mu}^{15}(\frac{1}{2}) = 100\, 011\, 011\, 011\, 011$, and it is, of course, a word in ${\cal L}_{15}$.  
If we set $\langle z \rangle_8 = \frac{\lfloor 2^8 z \rfloor}{2^8}$, meaning that round down the nearest to $2^{-9}$, then 
Algorithm~\ref{alg:calc-wrong} outputs $100\, 011\, 011\, 011\, 00$, which is not a word of ${\cal L}_{14}$. 
Similarly for the same $\mu$, 
if we set $\langle z \rangle = \frac{\lfloor 1000 z \rfloor}{1000}$, meaning that round down the nearest to $10^{-3}$, then 
Algorithm~\ref{alg:calc-wrong} outputs $100\, 011\, 011\, 011\, 010$, which is not a word of ${\cal L}_{15}$.  
\end{proof}

 Proposition~\ref{prop:round} might not be surprising. 
 Can we correct Algorithm~\ref{alg:calc-wrong} so as to output $b_1 \dots b_n \in {\cal L}_n$? 
 Yes it is possible if we set $\kappa = \Theta(n)$, that is sufficiently precise. 
 Clearly, it requires $\Theta(n)$ working space. 
 Then, it is validate to ask 
   {\em if it is possible to generate/recognize $b_1 \dots b_n \in {\cal L}_n$ in $\order(n)$}. 

 


\subsection{Problems and Results}\label{sec:result}
 The following problems could be natural in the sense of computational complexity of tent codes. 
\begin{problem}[Decision]\label{prob:dec}
 Given 
   a real $x \in [0,1)$ and a bit sequence $\bitseq_n \in \{0,1\}^n$, 
  decide  if $\bitseq_n = \enc^n(x)$. 
\end{problem}
\begin{problem}[Calculation]\label{prob:calc}
 Given a real $x \in [0,1)$ and a positive integer $n$, 
  find $\enc^n(x) \in \{0,1\}^n$. 
\end{problem}
 Recalling Proposition~\ref{prop:round}, 
  it seems difficult to solve Problems~\ref{prob:dec} and~\ref{prob:calc} exactly, in $\order(n)$ space.  
 Then, we consider to compute a {\em valid} tent code {\em around} $x$. 
 Let $0 < \epsilon \ll 1$, we define 
\begin{align*}
 {\cal L}_n(x,\epsilon) &= \{ \enc^n(x') \mid x-\epsilon \leq x' \leq x+\epsilon  \}
\end{align*}
  for $x \in [0,1)$. 
 It is equivalently rephrased by 
\begin{align}
 {\cal L}_n(x,\epsilon) &= \{ \bitseq \in {\cal L}_n \mid \enc^n(x-\epsilon) \preceq \bitseq \preceq \enc^n(x+\epsilon)  \} 
\label{eq:approxL}
\end{align}
  by Propositions \ref{prop:tent-expansion} and \ref{prop:order}. 
 For the calculation Problem~\ref{prob:calc},  we establish the following simple theorem.  
\begin{theorem}[Approximate calculation]\label{thm:calc}
  Let $\mu \in (1, 2)$ be rational given by an irreducible fraction $\mu = c/d$, and 
   let $0 < \epsilon <1/4$. 
 Given a real\footnote{By a read only tape of infinite length. } $x \in [0,1)$, 
  Algorithm~\ref{alg:calc-improve} described in Section~\ref{sec:calc} outputs $\bitseq_n \in {\cal L}_n(x,\epsilon)$. 
  The space complexity of Algorithm~\ref{alg:calc-improve} is $\Order(\lg^2 \epsilon^{-1} \lg d / \lg^2 \mu + \lg n)$. 
\end{theorem}
 Then, the following theorem for Problem~\ref{prob:dec} is the main result of the paper.  
\begin{theorem}[Decision for $\epsilon$-perturbed input]\label{thm:smooth-recog}
  Let $\mu \in (1, 2)$ be rational given by an irreducible fraction $\mu = c/d$, and 
   let $0 < \epsilon <1/4$. 
 Given a bit sequence $\bitseq_n \in \{0,1\}^n$ and 
  a real\footnote{By a read only tape of infinite length. } $x \in [0,1)$, 
 Algorithm~\ref{alg:recog} described in Section~\ref{sec:smooth-recog} accepts it if $\bitseq_n \in {\cal L}_n(x,\epsilon)$ and 
  rejects it if  $\bitseq_n \not \in {\cal L}_n(x,2\epsilon)$. 
 If an ($\epsilon$-perturbed) instance $\bitseq_n$ is given by $\bitseq_n = \enc^n(X)$ for $X \in [x-\epsilon,x+\epsilon]$ uniformly  at random 
  then the space complexity of Algorithm~\ref{alg:recog} is $\Order(\lg^2 n/\lg^3 d + \lg \epsilon^{-1} /\lg d)$ in expectation.  
\end{theorem} 
 As stated in theorems, 
  this paper assumes $\mu$ rational 
  mainly for the reason of Turing comparability, but it is not essential\footnote{
    We can establish some arguments for any real $\mu \in (0,1)$ 
     similar (but a bit weaker) to the theorems (see also \cite{OK23}). }.
 Instead, we allow an input instance $x \in [0,1)$ being a {\em real}\footnote{
  We do not use this fact directly in this paper, but 
   it might be worth to mention it for some conceivable variants 
     in the context of smoothed analysis to draw $X \in [x-\epsilon,x+\epsilon]$ uniformly at random. 
 }, given by a read only tape of infinite length.
 We remark that the space complexity of Theorem~\ref{thm:calc} is optimal in terms of $n$. 
 
\paragraph{Proof strategy of the theorems.}
 For proofs, 
   we will introduce the ``automaton'' for ${\cal L}_n$ given by \cite{OK23} in Section~\ref{sec:OK23}. 
 Once we get the automaton, 
  Theorem~\ref{thm:calc} is not difficult, 
   and we prove it in Section~\ref{sec:calc}. 
 Algorithm~\ref{alg:calc-improve} is relatively simple and 
   the space complexity is trivial. 
 Thus, the correctness is the issue, but it is also not very difficult. 
 Then, we give Algorithm~\ref{alg:recog} and prove Theorem~\ref{thm:smooth-recog}  in Section~\ref{sec:smooth-recog}. 
 The correctness is essentially the same as  Algorithm~\ref{alg:calc-improve}. 
 The analysis of the space complexity is the major issue. 

  

\section{Underlying Technology}\label{sec:OK23}
 This section 
   briefly
   introduces some fundamental technology for the analyses in Sections~\ref{sec:calc} and \ref{sec:smooth-recog} 
   including the automaton and the Markov chain for ${\cal L}_n$, according to \cite{OK23}. 

 The key idea of a space efficient computation of a tent code 
   is a representation of an equivalent class with respect to $\enc^n$. 
 Let
\begin{align}
  \typefn(\bitseq_n) &\defeq \{ \tent^{n}(x) \mid \enc^n(x)= \bitseq_n \} 
  \label{def:type}
\end{align}
 for $\bitseq_n \in {\cal L}_n$, 
 we call $\typefn(\bitseq_n)$ the {\em segment-type} of $\bitseq_n$. 
%
In fact, $\typefn(\bitseq_n)$ is a continuous interval, where one end is open and the other is close. 
By some straightforward argument with \eqref{def:encode2}, we get the following recursive formula. 
\begin{lemma}[\cite{OK23}]
\label{lem:transition}
Let $x \in [0,1)$, and let $\enc^n(x) = b_1\cdots b_n$. \\
 (1) Suppose $\typefn^i(x) = [v,u)$ ($v<u$). 
\vspace{-1ex}
\begin{itemize}\setlength{\itemindent}{4em}\setlength{\parskip}{0.5ex}\setlength{\itemsep}{0cm} 
\item[Case 1-1:]  $v < \frac{1}{2} < u$. 
\vspace{-1ex}
\begin{itemize}\setlength{\itemindent}{4em}\setlength{\parskip}{0.5ex}\setlength{\itemsep}{0cm} 
\item[Case 1-1-1.] If $\tent^i(x) < 1/2$ then $\typefn^{i+1}(x) = [\tent(v),\tent(\tfrac{1}{2}))$, and $b_{i+1}=0$. 
\item[Case 1-1-2.] If $\tent^i(x) \geq 1/2$ then $\typefn^{i+1}(x) =  (\tent(u),\tent(\tfrac{1}{2})]$, and $b_{i+1}=1$. 
\end{itemize}
\item[Case 1-2:] $u \leq \frac{1}{2}$. Then $\typefn^{i+1}(x) = [\tent(v),\tent(u))$, and $b_{i+1}=0$.  
\item[Case 1-3:] $v \geq \frac{1}{2}$. Then $\typefn^{i+1}(x) = (\tent(u),\tent(v)]$, and $b_{i+1}=1$. 
\end{itemize}
\vspace{-1ex}

\noindent (2) Similarly, suppose $\typefn^i(x) = (v,u]$ ($v<u$). 
\vspace{-1ex}
\begin{itemize}\setlength{\itemindent}{4em}\setlength{\parskip}{0.5ex}\setlength{\itemsep}{0cm} 
\item[Case 2-1:]  $v < \frac{1}{2} < u$. 
\vspace{-1ex}
\begin{itemize}\setlength{\itemindent}{4em}\setlength{\parskip}{0.5ex}\setlength{\itemsep}{0cm} 
\item[Case 2-1-1.] If $\tent^n(x) \leq 1/2$ then $\typefn^{i+1}(x) = (\tent(v),\tent(\tfrac{1}{2})]$, and $b_{i+1}=1$. 
\item[Case 2-1-2.] If $\tent^n(x) > 1/2$ then $\typefn^{i+1}(x) =  [\tent(u),\tent(\tfrac{1}{2}))$, and $b_{i+1}=0$. 
\end{itemize}
\item[Case 2-2:] $u \leq \frac{1}{2}$. Then $\typefn^{i+1}(x) = (\tent(v),\tent(u)]$, and $b_{i+1}=1$. 
\item[Case 2-3:] $v \geq \frac{1}{2}$. Then $\typefn^{i+1}(x) = [\tent(u),\tent(l))$, and $b_{i+1}=0$. 
\end{itemize}
\end{lemma} 

Let
\begin{equation}
  \typeset_n = \{\typefn(\bitseq) \subseteq [0,1)  \mid \bitseq \in {\cal L}_n \}
\label{def:typeset}
\end{equation}
denote the set of segment-types (of ${\cal L}_n$). 
 It is easy to observe that $|{\cal L}_n|$ can grow exponential to $n$, 
 while the following theorem implies that $|\typeset_n| \leq 2n$. 
\begin{theorem}[\cite{OK23}]
  \label{theo:number_of_type}
 Let $\mu \in (1, 2)$. 
 Let $\bitseqc_i = \enc^i(\frac{1}{2})$, and let  
 \begin{align}
 \typei_i = \typefn(\bitseqc_i) \hspace{2em}\mbox{and}\hspace{2em}
 \flip{\typei}_i = \typefn(\flip{\bitseqc}_i)
 \label{def:typei}
 \end{align}
  for $i=1,2,\ldots$. 
Then, 
  \begin{align*}
    \typeset_n  = \bigcup_{i=1}^{n_*} \left\{ \typei_i, \flip{\typei}_i \right\}
  \end{align*}
  for $n \geq 1$, 
 where 
  $n_* = \min (\{i \in \{1,2,\ldots,n-1\} \mid \typei_{i+1} \in \typeset_i\} \cup \{n\})$. 
\end{theorem}

The following lemma, derived from Lemma~\ref{lem:transition}, gives an explicit recursive formula of $\typei_i$ and $\flip{\typei}_i$. 
\begin{lemma}[\cite{OK23}]\label{lem:typei}
  $\typei_i$ and $\flip{\typei}_i$ given by \eqref{def:typei} 
   are recursively calculated for $i=1,2,\ldots$ as follows. 
 \begin{align*} 
 \typei_1 = [0,\tfrac{\mu}{2})\hspace{2em}\mbox{and}\hspace{2em}
 \flip{\typei}_1 = (0,\tfrac{\mu}{2}].
 \end{align*}
For $i=2,3,\ldots$,
 \begin{align} 
 \typei_{i} &= \begin{cases}
  \begin{cases}
  [\tent(v),\tent(\tfrac{1}{2})) &: v < \tfrac{1}{2} < u \\ 
  [\tent(v),\tent(u))  &:  u \leq \tfrac{1}{2}  \\
  (\tent(u),\tent(v)]  &:  v \geq \tfrac{1}{2} 
  \end{cases}
 &\mbox{if $\typei_{i-1} = [v,u)$, } \\
 \begin{cases} 
 [\tent(u),\tent(\tfrac{1}{2})) &: v < \tfrac{1}{2} < u \\ 
 (\tent(v),\tent(u)]  &: u \leq \tfrac{1}{2}  \\
 [\tent(u),\tent(v))  &: v \geq \tfrac{1}{2}  
 \end{cases}
&\mbox{if $\typei_{i-1} = (v,u]$} 
\end{cases}
\label{eq:20231105a}
\end{align} 
holds. Then, 
\begin{align*} 
\flip{\typei}_{i} &= \begin{cases}
(v,u] & \mbox{if $\typei_{i} = [v,u)$} \\
[v,u) & \mbox{if $\typei_{i} = (v,u]$} 
\end{cases}
\end{align*} 
holds. 
\end{lemma}

 For convenience, 
 we define the {\em level} of  $\typej \in \typeset_n$ 
  by 
\begin{align}
  L(\typej) = k
\label{def:level}
\end{align}
   if $\typej = \typei_k$ or $\flip{\typei_k}$.  
 Notice that Theorem~\ref{theo:number_of_type} implies that 
   the level of $\typefn(\bitseq_k)$ for $\bitseq_k \in {\cal L}_k$ may be strictly less than $k$. 
 In fact, it happens, which provides a {\em space efficient} ``automaton''. 

\begin{figure}[tb]
  \centering
  \begin{tikzpicture}
    [node distance=.5cm, every loop/.style={looseness=2}]
    \node [state, initial left] (v0) [] {$q_{0}$};
    \node [state] (a1) [above right = of v0] {$\typei_{1}$};
    \node [state] (a2) [right = of a1] {$\typei_{2}$};
    \node [state] (a3) [right = of a2] {$\typei_{3}$};
    \node [state] (a4) [right = of a3] {$\typei_{4}$};
    \node [state] (a5) [right = of a4] {$\typei_{5}$};
    \node [state] (a6) [right = of a5] {$\typei_{6}$};
    \node [state, draw=none] (ai) [right = of a6] {$\cdots$};
    \node [state] (an) [right = of ai] {$\typei_{n}$};
    \node [state] (b1) [below right = of v0] {$\flip{\typei}_{1}$};
    \node [state] (b2) [right = of b1] {$\flip{\typei}_{2}$};
    \node [state] (b3) [right = of b2] {$\flip{\typei}_{3}$};
    \node [state] (b4) [right = of b3] {$\flip{\typei}_{4}$};
    \node [state] (b5) [right = of b4] {$\flip{\typei}_{5}$};
    \node [state] (b6) [right = of b5] {$\flip{\typei}_{6}$};
    \node [state, draw=none] (bi) [right = of b6] {$\cdots$};
    \node [state] (bn) [right = of bi] {$\flip{\typei}_{n}$};
    \path [->] (v0) edge [] node [above left] {1} (a1);
    \path [->] (a1) edge [] node [above] {0} (a2);
    \path [->] (a2) edge [] node [above] {0} (a3);
    \path [->] (a3) edge [] node [above] {1} (a4);
    \path [->] (a4) edge [] node [above] {0} (a5);
    \path [->] (a5) edge [] node [above] {0} (a6);
    \path [->] (a6) edge [] node [above] {} (ai);
    \path [->] (ai) edge [] node [above] {} (an);
    \path [->] (a1) edge [loop above] node [above] {1} ();
    \path [->] (a2) edge [bend left] node [pos=.5, right] {1} (b2);
    \path [->] (a4) edge [bend right=3] node [pos=.1, left] {1} (b3);
    \path [->] (a5) edge [bend left=3] node [pos=.1, below] {1} (b2);
    \path [->] (v0) edge [] node [below left] {0} (b1);
    \path [->] (b1) edge [] node [below] {1} (b2);
    \path [->] (b2) edge [] node [below] {1} (b3);
    \path [->] (b3) edge [] node [below] {0} (b4);
    \path [->] (b4) edge [] node [below] {1} (b5);
    \path [->] (b5) edge [] node [below] {1} (b6);
    \path [->] (b6) edge [] node [below] {} (bi);
    \path [->] (bi) edge [] node [below] {} (bn);
    \path [->] (b1) edge [loop below] node [below] {0} ();
    \path [->] (b2) edge [bend left] node [pos=.5, left] {0} (a2);
    \path [->] (b4) edge [bend left=3] node [pos=.1, left] {0} (a3);
    \path [->] (b5) edge [bend right=3] node [pos=.1, above] {0} (a2);
  \end{tikzpicture}
  \caption{Transition diagram over $\typeset_n$ for $\mu=1.6$. }
  \label{fig:stg1}
\end{figure}

\paragraph{State transit machine (``automaton''). }
 By Theorem~\ref{theo:number_of_type}, 
  we can design a {\em space efficient} state transit machine\footnote{
   Precisely, we need a ``counter'' for the length $n$ of the string, 
    while notice that our main goal is not to design an automaton for ${\cal L}_n$. 
   Our main target Theorems~\ref{thm:calc} and \ref{thm:smooth-recog} assume a standard Turing machine, where 
   obviously  we can count the length $n$ of a sequence in $\Order(\log n)$ space.  
   } according to Lemma~\ref{lem:transition}, 
   to recognize ${\cal L}_n$. 
 We define the set of states by 
  $Q_n = \{q_0\} \cup \{\emptyset\} \cup \typeset_n$, 
     where $q_0$ is the initial state, and $\emptyset$ denotes the unique reject state. 
 Let $\delta\colon Q_{n-1} \times \{0,1\} \to Q_n$ denote the state transition function, defined as follows.  
 Let $\delta(q_0,1) = \typei_1$ and $\delta(q_0,0) = \flip{\typei}_1$. 
 According to Lemma~\ref{lem:transition}, we appropriately define  
\begin{align}
 \delta(\typej,b) = \typej'
\label{def:delta}
\end{align}
 for $J \in \typeset_{n-1}$ and $b \in \{0,1\}$, 
as far as $J$ and $b$ are consistent. 
 If the pair $\typej$ and $b$ are inconsistent,  
  we define $\delta(\typej,b)=\emptyset$; 
  precisely 
\begin{align*}\begin{cases}
 \mbox{$\typej = (v,u]$ and $v \geq \frac{1}{2}$}  &\mbox{(cf. Case 1-3)}   \\
 \mbox{$\typej = [v,u)$ and $u \leq \frac{1}{2}$} &\mbox{(cf. Case 2-2)}   \\
 \mbox{$\typej = [v,u)$ and $u \leq \frac{1}{2}$} &\mbox{(cf. Case 1-2)}   \\
 \mbox{$\typej = (v,u]$ and $v \geq \frac{1}{2}$} &\mbox{(cf. Case 2-3)}   
\end{cases}
\end{align*} 
are the cases, where $v=\inf \typej$ and  $u=\sup \typej$. 
\begin{lemma}[\cite{OK23}]
  \label{lemm:space_of_stg}
  Let $\mu \in (1, 2)$ be a rational given by an irreducible fraction $\mu = c/d$.
  For any $k \in \mathbb{Z}_{>0}$,
  the state transit machine on $Q_k$ is represented by $\Order(k^{2} \lg{d})$ bits. 
\end{lemma}

 We will use the following two technical lemmas about the transition function in Sections~\ref{sec:calc} and~\ref{sec:smooth-recog}. 
\begin{lemma}[\cite{OK23}]\label{lemm:transition_function}
 Suppose for $\mu \in (1, 2)$ that $\tent_{\mu}^{i}(\frac{1}{2}) \neq \frac{1}{2}$ holds for any $i = 1, \dots, n-1$. 
 Then, 
 \begin{align*}
   \delta(\typei_{n}, b) &\in \left\{\typei_{n+1}\right\} \cup \left\{\flip{\typei}_{k+1} \mid 1 \le k \le \tfrac{n}{2} \right\} \cup \{\emptyset\}
 \end{align*}
   hold for $b=0,1$. 
\end{lemma}
\begin{lemma}[\cite{OK23}]\label{lem:back2}
 Suppose for $\mu \in (1, 2)$ that $\tent_{\mu}^{i}(\frac{1}{2}) \neq \frac{1}{2}$ holds for any $i = 1, \dots, 2n-1$. 
 Then, there exists $k \in \{n+1,\ldots,2n\}$ and $b \in \{0,1\}$ such that 
 \begin{align*}
   \delta(\typei_{k}, b) &\in \left\{\flip{\typei}_{k'+1} \mid 1 \le k' \le \tfrac{k}{2} \right\} 
 \end{align*}
   hold. 
\end{lemma}
 Roughly speaking, Lemma~\ref{lemm:transition_function} implies that 
  the level increases by one, or decreases into (almost) a half by a transition step. 
 Furthermore, 
  Lemma~\ref{lem:back2} implies that 
  there is at least one way to decrease the level during $n,\ldots,2n$.

\paragraph{Markov model.}
 Furthermore, 
  the state transitions preserve the uniform measure, over $[0,1)$ in the beginning, 
  since the tent map is  piecewise linear. 
\begin{lemma}[\cite{OK23}]\label{lem:cond-prob}
Let $X$ be a random variable drawn from $[0,1)$ uniformly at random. 
Let $\bitseq_n \in {\cal L}_n$. Then, 
\begin{align*}
 \Pr[ B_{n+1}= b \mid \enc^n(X) = \bitseq_n]
    = \frac{|\typefn(\bitseq_n b)|}{|\typefn(\bitseq_n 0)|+|\typefn(\bitseq_n 1)|}
\end{align*}
 holds for $b \in \{0,1\}$, where let  $|\typefn(\bitseq_n b)| = 0$ if $\bitseq_n b \not\in {\cal L}_{n+1}$. 
\end{lemma}
 Let ${\cal D}_{n,\mu}$ (or simply ${\cal D}_n$) denote a probability distribution over ${\cal L}_n$ 
  which follows $\enc^{n}(X) $ for $X$ is uniformly distributed over $[0,1)$, 
 i.e., ${\cal D}_n$ represents the probability of appearing $\bitseq_n \in {\cal L}_n$ 
   as given the initial condition $x$ uniformly at random. 
%
\begin{theorem}[\cite{OK23}]\label{thm:main}
 Let $\mu \in (1,2)$ be a rational given by an irreducible fraction $\mu=c/d$. 
 Then, it is  possible to generate $B \in {\cal L}_n$ according to ${\cal D}_n$ in $\Order(\lg^2 n \lg^3 d / \lg^4 \mu)$ space in expectation 
 (as well as, with high probability). 
\end{theorem}
Thus, we remark that the tent language ${\cal L}_n$ is recognized in $\Order(\lg^2 n \lg^3 d / \lg^4 \mu)$ space
{\em  on average} all over the initial condition $x \in [0,1)$, by Theorem~\ref{thm:main}.

\section{Calculation in ``Constant'' Space}\label{sec:calc}
 Theorem~\ref{thm:calc} is easy, once the argument in Section~\ref{sec:OK23} is accepted. 
 Algorithm~\ref{alg:calc-improve} shows the approximate calculation of $\enc^n(x)$ 
    so that the output is a {\em valid} tent code (recall the issue in Section~\ref{sec:issue}). 
%
%
 Roughly speaking, 
  Algorithm~\ref{alg:calc-improve} 
    calculates by rounding-off to $\kappa = \Order(\log \epsilon^{-1})$ bits 
    for the first $\kappa$ iterations (lines 5--10), and 
   then traces the automaton within the level $2 \kappa$ after the $\kappa$-th iteration (lines 11--20). 
%
 The algorithm traces finite automaton, and the desired space complexity is almost trivial. 
 A main issue is the correctness; 
   it is also trivial that the output sequence $\bitseq_n \in \{0,1\}^n$ is a valid tent code, 
  then 
   our goal is to prove 
   $\enc^n(x-\epsilon) \preceq \bitseq_n \preceq \enc^n(x+\epsilon)$.  
 The trick is based on the fact that the tent map is an extension map, 
  which will be used again in the smoothed analysis in the next session.

Then we explain the detail of Algorithm~\ref{alg:calc-improve}. 
 Let $\langle x \rangle_k$ denote a binary expression by the rounding off a real $x$ to the nearest $1/2^{k+1}$, 
  where the following argument only requires $|\langle x \rangle_k - x| \leq 1/2^k$, 
  meaning that rounding up and down is not essential.  
 Naturally assume that 
   $\langle x \rangle_k$ for $x \in [0,1)$ is represented by $k$ bits, 
   meaning that the space complexity of $\langle x \rangle_k$ is $\Order(k)$. 

 In the algorithm, 
  rationals $v[k]$ and $u[k]$ respectively denote  $\inf \typei_k$ and $\sup \typei_k$ for $k=1,2,\ldots$. 
 For descriptive purposes, 
  $v[0]$ and $u[0]$ corresponds to $q_0$, and  $v[-1]$ and $u[-1]$ corresponds to the reject state $\emptyset$ in $Q_n$. 
 The single bit $c[k]$ denotes $c_k$ for $\bitseqc_n =c_1 \cdots c_n = \enc^n(\frac{1}{2})$ (recall Theorem~\ref{theo:number_of_type}). 
 Thus, $\typei_k = [v[k],u[k])$ and $\flip{\typei}_k =(v[k],u[k]]$ if $c[k]=0$, 
  otherwise $\typei_k = (v[k],u[k]]$ and $\flip{\typei}_k =[v[k],u[k])$ see Section~\ref{sec:OK23}. 
 The integer $\delta(l,b) = l'$ represents the transition $\delta(\typei_l,b) = \typej'$, 
   where $\typej' = \typei_{l'}$ or $\flip{\typei}_{l'}$, 
  given by \eqref{def:delta} in Section~\ref{sec:OK23}.  
 Notice that if $\delta(\typei_l,b) = \typej'$ then  $\delta(\flip{\typei}_l,b) = \flip{\typej'}$ holds \cite{OK23}. 
 The pair $l$ and $b$ represent $Z_i = \typei_l$ if $b=c[l]$, 
  otherwise, i.e., $\flip{b} =c[l]$, $Z_i = \flip{\typei}_l$,
  at the $i$-th iteration (for $i=1,\ldots n$). 
 See Section~\ref{apx:subpro} for the detail of the subprocesses, Algorithms~\ref{alg:vucdelt} and \ref{alg:computel}.

\begin{algorithm}[t]
    \caption{Valid calculation with ``constant'' space (for $\mu,\epsilon$)}
    \label{alg:calc-improve}
    \begin{algorithmic}[1]
    \REQUIRE a real $x \in [0,1]$ 
    \ENSURE a bit sequence $b_1\cdots b_n \in {\cal L}_n(x,\epsilon)$ 
    \STATE {\rm int} $\kappa \leftarrow \lceil -3\lg \epsilon /\lg \mu \rceil$
    \STATE {\bf compute} rational $v[k]$, rational $u[k]$, bit $c[k]$, int $\delta[k,0]$ and int $\delta[k,1]$ for $k=-1$ to $2\kappa$ 
      by Algorithm~\ref{alg:vucdelt}
    \STATE {\rm int} $l \leftarrow 0$, {\rm bit} $b \leftarrow 1$
    \STATE {\rm rational} $z \leftarrow \langle x \rangle_{\kappa}$
    \FOR{$i=1$ to $\kappa$}
    \STATE  {\bf compute} {\rm bit} $b'$ based on $b$ and $z$ by \eqref{def:encode2}
    \RETURN $b'$ \COMMENT{as $b_i$}
    \STATE  {\bf update} $l$ based on  $b$ and $b'$, {\bf update} $b$ (by Algorithm~\ref{alg:computel})
    \STATE $z \leftarrow \langle \tent(z) \rangle_{\kappa}$
    \ENDFOR
    \FOR{$i=\kappa+1$ to $n$}
    \STATE $b' \leftarrow \arg \min \{ \delta(l,b'')  \mid b'' \in \{0,1\},\ \delta(l,b'') > 0 \}$
    \IF{$b=c[l]$}
    \STATE $b \leftarrow b'$
    \ELSE
    \STATE $b \leftarrow \flip{b'}$
    \ENDIF
    \RETURN $b$ \COMMENT{as $b_i$}
    \STATE $l \leftarrow \delta[l,b']$
    \ENDFOR
    \end{algorithmic}
\end{algorithm}

\begin{theorem}[Theorem~\ref{thm:calc}]
 Given a real $x \in [0,1)$, 
  Algorithm~\ref{alg:calc-improve} outputs $\bitseq_n \in {\cal L}_n(x,\epsilon)$. 
  The space complexity of Algorithm~\ref{alg:calc-improve} is $\Order(\lg^2 \epsilon^{-1} \lg d / \lg^2 \mu + \lg n)$. 
\end{theorem}
\begin{proof}
 To begin with, 
    we remark that  Algorithm~\ref{alg:calc-improve} constructs 
    the transition diagram only up to the level $2\kappa$. 
 Nevertheless,  
   Algorithm~\ref{alg:calc-improve} correctly conducts the for-loop in lines 11--20, 
    meaning that all transitions at line 19 are always valid during the loop. 
 This is due to Lemma~\ref{lem:back2}, 
   which claims that there exists at least one $l \in \{\kappa + 1,\ldots,2 \kappa\}$ 
     such that $\delta(l,b) \leq \kappa$ (and also $\delta(l,b) >0$), meaning that it is a reverse edge. 
 Then, it is easy from Lemma~\ref{lemm:space_of_stg} that  
  the space complexity of Algorithm~\ref{alg:calc-improve} is $\Order(\kappa^2 \lg d) = \Order(\lg^2 \epsilon \lg d / \lg^2 \mu)$, 
  except for the space $\Order(\lg n)$ of counter $i$ (and $n$) at line 11. 

 Then, we prove $b_1 \cdots b_n \in {\cal L}_n(x,\epsilon)$. 
 Since the algorithm follows the transition diagram of ${\cal L}_n$ (recall Section~\ref{sec:OK23}), 
   it is easy to see that $b_1 \cdots b_n \in {\cal L}_n$, and hence  
  we only need to prove 
\begin{align}
 \enc^n(x-\epsilon) \preceq b_1 \cdots b_n \preceq \enc^n(x+\epsilon)
 \label{eq:round-order}
\end{align}
  holds. 
 We here only prove $\enc^n(x-\epsilon) \preceq b_1 \cdots b_n$ 
   while  $b_1 \cdots b_n \preceq \enc^n(x+\epsilon)$ is essentially the same. 
 The proof is similar to that of Proposition~\ref{prop:order}, while 
  the major issue is whether rounding $\langle z \rangle_{\kappa}$ preserves the ``ordering,'' 
  so does the exact calculation by Lemma~\ref{lem:encode1}. 

For convenience, 
 let $z_i$ denote the value of $z$ in the $i$-th iteration of the algorithm, i.e., 
 $z_{i+1} = \langle \tent(z_i) \rangle_{\kappa}$.  
Let $x^- = x-\epsilon/2$ and let $x^{--} = x-\epsilon$. 
The proof consists of two parts: 
\begin{align}
 \enc^{\kappa}(x^{--}) &\prec \enc^{\kappa}(x^-) \hspace{2em}\mbox{and} \label{eq:20231031a}\\
 \enc^{\kappa}(x^-) &\preceq b_1 \cdots b_{\kappa}. \label{eq:20231031b}
\end{align}

 For the claim \eqref{eq:20231031a}, notice that $\enc^{\kappa}(x^{--}) \preceq \enc^{\kappa}(x^-)$ by Lemma~\ref{lem:encode1}. 
 If $\enc^{\kappa}(x^{--}) = \enc^{\kappa}(x^-)$,  
  Lemma~\ref{lem:encode3} implies that 
   $|\tent^{\kappa}(x^{--}) - \tent^{\kappa}(x^-)| 
    = \frac{\epsilon}{2} \mu^{\kappa} 
    \geq \frac{\epsilon}{2} \mu^{-3 \lg \epsilon / \lg \mu} 
    = \frac{1}{2\epsilon^2 } > 1$, 
    which contradicts to $0 \leq \tent^{\kappa}(x') \leq 1$ for any $x' \in [0,1)$. 
 Now we get \eqref{eq:20231031a}. 

 For \eqref{eq:20231031b}, similar to the proof of Proposition~\ref{prop:order} based on Lemma~\ref{lem:encode1}, 
   we claim if $\enc^i(x^-) = b_1\cdots b_i$ then 
  \begin{align}
  \begin{cases}
    x^-_i < z_i  & \mbox{if $b_i = 0$, }\\ 
    x^-_i > z_i & \mbox{if $b_i  = 1$}
 \end{cases}
  \end{align}
  hold. 
 The basic argument is essentially the same as the proof of Lemma~\ref{lem:encode1}, 
  which is based on the argument of segment type (see \cite{OK23}), and 
  here it is enough to check 
\begin{align}
\mbox{ if $|x^-_i - z_i| \geq \frac{\epsilon}{2}$ then $|x^-_{i+1} - z_{i+1}| \geq \frac{\epsilon}{2}$, as far as $\enc^{i+1} (x) = \enc^{i+1} (z)$ }
\label{eq:20231030a}
\end{align}
  meaning that the rounding-off does not disrupt the order. 
Notice that 
\begin{align}
|x^-_{i+1} - z_{i+1}| 
 &= |\tent(x^-_i) - \langle \tent(z_i) \rangle_{\kappa} | 
 \geq |\tent(x^-_i) - \tent(z_i) | - |\langle \tent(z_i) \rangle_{\kappa} - \tent(z_i)| 
\label{eq:20231029a}
\end{align} 
 holds,  where the last inequality follows the triangle inequality
  $|\tent(x^-_i) - \langle \tent(z_i) \rangle_{\kappa} |  + |\langle \tent(z_i) \rangle_{\kappa} - \tent(z_i)| \geq |\tent(x^-_i) - \tent(z_i) |$. 
  Note that $|\tent(x^-_i) - \tent(z_i) | = \mu|x^-_i-z_i| \geq \mu\frac{\epsilon}{2}$ holds under the hypothesis $\enc^{i+1} (x) = \enc^{i+1} (z)$.  
 We also remark $|\langle \tent(z_i) \rangle_{\kappa} - \tent(z_i)| \leq \frac{1}{2^\kappa}$ by definition of 
  $\langle \cdot \rangle_{\kappa}$. 
 Furthermore, 
   we claim $\frac{1}{2^\kappa} \leq (\mu-1)\frac{\epsilon}{2}$ 
  by $\kappa \geq -3 \lg \epsilon / \lg \mu $:  
 note that $1 / \lg \mu \geq 0.7 - \lg (\mu -1)$ holds for $1 < \mu <2$, 
  and then  
  $\kappa \geq - 3 \lg \epsilon (0.7-\lg (\mu-1)) 
    = - 2.1\lg \epsilon + 3\lg \epsilon \lg (\mu-1) 
    \geq - \lg \epsilon^2 - \lg (\mu-1)
    \geq - \lg \frac{\epsilon}{2} - \lg (\mu-1)$ 
  holds
    where we use $\epsilon < 1/4$, 
  which implies the desired claim $\frac{1}{2^{\kappa}} \leq (\mu-1) \frac{\epsilon}{2}$. 
 Then, 
\begin{align}
 \eqref{eq:20231029a} 
 \geq \mu\tfrac{\epsilon}{2} - \tfrac{1}{2^{\kappa}} 
 \geq \mu\tfrac{\epsilon}{2} - (\mu-1) \tfrac{\epsilon}{2} = \tfrac{\epsilon}{2}, 
\end{align} 
 and we got \eqref{eq:20231030a}, and hence \eqref{eq:20231031b} by Proposition~\ref{prop:order}. 
 By \eqref{eq:20231031a} and \eqref{eq:20231031b}, 
   $\enc^{\kappa}(x-\epsilon) \preceq b_1 \cdots b_{\kappa}$ is easy. 
 Now, 
   we obtain \eqref{eq:round-order} by Proposition~\ref{prop:order}.  
\end{proof}

\section{Smoothed Analysis for Decision}\label{sec:smooth-recog}
 Now, we are concerned with the decision problem, Problem~\ref{prob:dec}. 
%
 Algorithm~\ref{alg:recog} efficiently solves the problem with $\epsilon$-perturbed input, for Theorem~\ref{thm:smooth-recog}. 
Roughly speaking, 
 Algorithm~\ref{alg:recog} 
  checks whether $\bitseq_n \in {\cal L}_n$ at line 24, and 
  checks whether $\enc^n(x-\epsilon) \preceq \bitseq_n \preceq \enc^n(x+\epsilon)$ for lines 6--22. 
 Lines 25--27 show a deferred update of those parameters, 
  to save the space complexity. 

\begin{algorithm}[t]
    \caption{Decision (for $\mu,\epsilon$)}
    \label{alg:recog}
    \begin{algorithmic}[1]
    \REQUIRE a bit sequence $b_1 \cdots b_n \in \{0,1\}^n$ and a real $x \in [0,1]$ 
    \ENSURE Accept if $b_1\cdots b_n \in {\cal L}_n(x,\epsilon)$ and Reject if $b_1\cdots b_n \not\in {\cal L}_n(x,2\epsilon)$ 
    \STATE {\rm int} $\kappa \leftarrow \lceil -3\lg \epsilon/\lg \mu \rceil$
    \STATE {\bf compute} {\rm rational} $v[k]$, {\rm rational} $u[k]$, {\rm bit} $c[k]$, {\rm int} $\delta[k,0]$ and {\rm int} $\delta[k,1]$ for $k=-1$ to $\kappa$ by Algorithm~\ref{alg:vucdelt}
    \STATE {\rm int} $k \leftarrow \kappa$, {\rm int} $l \leftarrow 0$, {\rm int} ${\it case} \leftarrow 1$
    \STATE {\rm rational} $z^- \leftarrow \langle x-\frac{3}{2}\epsilon \rangle_{\kappa}$, 
                {\rm rational} $z^+ \leftarrow \langle x+\frac{3}{2}\epsilon \rangle_{\kappa}$
    \FOR{$i=1$ to $n$}
    \IF{${\it case}=1$}
    \STATE {\bf compute} {\rm bit} $b^-$, $b^+$ respectively based on $z^-$, $z^+$ with $b_{i-1}$ by \eqref{def:encode2}
    \STATE {\bf if} $b_i < b^-$ or $b^+ < b_i$ {\bf then} {\bf return} Reject and {\bf halt}
    \STATE {\bf else if} $b^- < b_i$ and $b_i = b^+ $ {\bf then} ${\rm case} \leftarrow 2$
    \STATE {\bf else if} $b^- = b_i$ and $b_i < b^+ $ {\bf then} ${\rm case} \leftarrow 3$
    \STATE {\bf else} $z^- \leftarrow \langle \tent(z^-) \rangle_{\kappa}$, $z^+ \leftarrow \langle \tent(z^+) \rangle_{\kappa}$ \COMMENT{i.e., $b_- = b_+ = b_i$}
    \ELSIF{${\it case}=2$}
    \STATE {\bf compute} {\rm bit} $b^+$ based on $z^+$ with $b_{i-1}$ by \eqref{def:encode2}
    \STATE {\bf if} $b^+ < b_i$ {\bf then} {\bf return} Reject and {\bf halt}
    \STATE {\bf else if} $b_i < b^+ $ {\bf then} ${\rm case} \leftarrow 0$
    \STATE {\bf else} $z^+ \leftarrow \langle \tent(z^+) \rangle_{\kappa}$ \COMMENT{i.e., $b_+ = b_i$}
    \ELSIF{${\it case}=3$}
    \STATE {\bf compute} {\rm bit} $b^-$ based on $z^-$  with $b_{i-1}$ by \eqref{def:encode2}
    \STATE {\bf if} $b_i < b^-$ {\bf then} {\bf return} Reject and {\bf halt}
    \STATE {\bf else if} $b^- < b_i$ {\bf then} ${\rm case} \leftarrow 0$
    \STATE {\bf else} $z^- \leftarrow \langle \tent(z^-) \rangle_{\kappa}$ \COMMENT{i.e., $b_- = b_i$}
    \ENDIF
    \STATE  {\bf update} $l$ based on  $b_{i-1}$ and $b_i$ (by Algorithm~\ref{alg:computel})
    \STATE {\bf if} $l = -1$ {\bf then} {\bf return} Reject and {\bf halt}
    \IF{$l=k$}
    \STATE {\bf compute} $v[k+1]$, $u[k+1]$, $c[k+1]$, $\delta[k,0]$ and $\delta[k,1]$ by Algorithm~\ref{alg:vucdelt}
    \ENDIF
    \ENDFOR
    \RETURN Accept
    \end{algorithmic}
\end{algorithm}

\begin{theorem}[Theorem~\ref{thm:smooth-recog}]
 Given a bit sequence $\bitseq_n \in \{0,1\}^n$ and 
  a real $x \in [0,1)$, 
 Algorithm~\ref{alg:recog} accepts it if $\bitseq_n \in {\cal L}(x,\epsilon)$ and 
  rejects it if  $\bitseq_n \not \in {\cal L}_n(x,2\epsilon)$. 
 If an ($\epsilon$-perturbed) instance $\bitseq_n$ is given by $\bitseq_n = \enc^n(X)$ for $X \in [x-\epsilon,x+\epsilon]$ uniformly  at random 
  then the space complexity of Algorithm~\ref{alg:recog} is $\Order(\lg^2 n/\lg^3 d + \lg \epsilon^{-1} /\lg d)$ in expectation.  
\end{theorem}
\begin{proof}
 The correctness proof is essentially the same as that of Theorem~\ref{thm:calc}. 
 In the algorithm, line 29 checks whether $b_1\cdots b_n \in {\cal L}_n$. 
 We see whether  
    $\enc^n(x-\epsilon) \prec b_1 \cdots b_n \prec \enc^n(x+\epsilon)$ 
   for the first at most $\kappa$ iterations, as follows. 
 Let $b^-_i$ and $b^+_i$ respectively represent $b^-$ and $b^+$ computed at line 7, 13 or 18 of the $i$th iteration. 
 Then, $b^-_1 \cdots b^-_{\kappa} \prec \enc^{\kappa}(x) \prec b^+_1 \cdots b^+_{\kappa}$ hold, 
  by the essentially same way as \eqref{eq:20231031a} in the proof of  Theorem~\ref{thm:calc}. 
 Similarly, 
  $\enc^{\kappa}(x-\epsilon) \prec b^-_1 \cdots b^-_{\kappa}$ holds, $b^+_1 \cdots b^+_{\kappa}\prec \enc^{\kappa}(x+\epsilon)$ as well. 
 Thus, 
  $\bitseq_n \prec \enc^{\kappa}(x-\epsilon)$ is safely rejected at line 8 or 19, 
  $\enc^{\kappa}(x+\epsilon) \prec \bitseq_n $ as well at line 8 or 13. 
 Thus we obtain the desired decision.

 Then we are concerned with the space complexity. 
 The analysis technique is very similar to or essentially the same as \cite{OK23} for a random generation of ${\cal L}_n$. 
 Let $X$ be a random variable drawn from the interval $[x-\epsilon,x+\epsilon]$ uniformly at random. 
 Let $\enc^n(X) = B_1,$
 Let  
\begin{align}
  \requiredstgsize=\max\{ k \in \mathbb{Z}_{>0} \mid L(\typefn(\enc^i(X))) = k \}
\label{def:K}
\end{align}  
  be a random variable 
  where $L(\typej) = k$ if $\typej=\typei_k$ or $\flip{\typei}_k$  
   (recall \eqref{def:level} as well as Theorem~\ref{theo:number_of_type}). 
 Lemma \ref{lemm:space_of_stg} implies 
   that its space complexity is $\Order(\requiredstgsize^{2} \lg{d})$. 
 Lemma \ref{lem:EK2},  appearing below, implies 
  \begin{align*}
    \E[\Order(\requiredstgsize^{2} \lg{d})]
    &= \Order(\E[\requiredstgsize^{2}] \lg{d}) \\
    &= \Order(\lg^{2}{n} \lg^{3}{d} / \lg^4 \mu)  
  \end{align*}
  and we obtain the claim. 
\end{proof}
\begin{lemma}
  \label{lem:EK2}
  Let $\mu \in (1, 2)$ be rational given by an irreducible fraction $\mu = c/d$.
  Suppose for $\mu \in (1, 2)$ that $\tent_{\mu}^{i}(\frac{1}{2}) \neq \frac{1}{2}$ holds for any $i = 1, \dots, n-1$. 
  Then,
    $\E[\requiredstgsize^{2}] = \Order(\log_{\mu}^{2}{n} \log_{\mu}^{2}{d}) = \Order(\lg^2 n \lg ^2 d / \lg^4 \mu)$. 
\end{lemma}
 We remark that the assumption of rational $\mu$ is not essential in Lemma~\ref{lem:EK2}; 
  the assumption is just for an argument about Turing comparability. 
 We can establish a similar (but a bit weaker) version of Lemma~\ref{lemm:average_space_complexity}
  for any real $\mu \in (0,1)$ (cf.\ Proposition 5.1 of \cite{OK23}). 
 Lemma~\ref{lem:EK2} is similar to Lemma~4.3 of \cite{OK23} for random generation, 
  where the major difference is that \cite{OK23} assumes $X_0$ is uniform on $[0,1)$ 
  while Lemma~\ref{lem:EK2} here assumes $X_0$ is uniform on $[x-\epsilon,x+\epsilon]$.  
 We only need to take care of some trouble 
  when the initial condition is around the boundaries of the interval, $x-\epsilon$ and $x+\epsilon$.

Suppose for the proof of Lemma~\ref{lem:EK2} that a random variable $X$ is drawn from the interval $[x-\epsilon,x+\epsilon]$ uniformly at random. 
Let $X_0 = X$ and let $X_i = \tent(X_{i-1})$  for $i=1,2,\ldots,n$. 
For convenience, 
 let $\typefn^i(x)$ denote $\typefn(\enc^i(x))$.  
Let $y \in [x-\epsilon,x+\epsilon]$. 
We say $X$ {\em covers around} $y$ at $i$-th iteration ($i \in \{1,2,\ldots,n\}$) if 
\begin{align}
\{\tent^i(y') \mid \enc^i(y') = \enc^i(y),\, x-\epsilon \leq y' \leq x+\epsilon\} = \typefn^i(y)
\end{align}
holds (recall $\typefn^i(y)=\{\tent^i(y') \mid \enc^i(y') = \enc^i(y)\}$ by definition \ref{def:type}). 
Similarly, 
 we say $X$ {\em fully covers} $S$ ($S \subseteq [x-\epsilon,x+\epsilon]$) at $i \in \{1,2,\ldots,n\}$ 
 if $X$ covers around every $y \in S$ at $i$. 
\begin{lemma}\label{lem:cover}
$X$ fully covers $(x-\epsilon + \frac{1}{n^2}, x + \epsilon - \frac{1}{n^2})$ at or after $\lceil 2\log_{\mu} n \rceil$ iterations. 
\end{lemma}
\begin{proof}
Let $k=\lceil 2\log_{\mu} n \rceil$. 
By Lemma~\ref{lem:encode3}, 
 if $y,y' \in [0,1)$ satisfies $\enc^k(y) = \enc^k(y')$ then 
   $|\tent^k(y) - \tent^k(y')| =\mu^k |y-y'| \geq \mu^{2\log_{\mu}n} |y-y'| = n^2 |y-y'|$. 
 On the other hand, $|\tent^k(y) - \tent^k(y')| \leq 1$, and hence
 the claim is easy from Propositions \ref{prop:tent-expansion} and~\ref{prop:order}. 
\end{proof}

Then, the proof of Lemma~\ref{lem:EK2} consists of two parts: 
 one is that 
 the conditional expectation of $K^2$ is $\Order(\lg^2 n \lg ^2 d / \lg^4 \mu)$ 
 on condition that $X \in (x-\epsilon + \frac{1}{n^2}, x + \epsilon - \frac{1}{n^2})$ (Lemma~\ref{lemm:average_space_complexity}),  
 and the other is that the probability of $X \not\in (x-\epsilon + \frac{1}{n^2}, x + \epsilon - \frac{1}{n^2})$ 
 is small enough to allow  Lemma~\ref{lem:EK2}  
  from Lemma~\ref{lemm:average_space_complexity}. 
 The latter claim is almost trivial (see the proof of   Lemma~\ref{lem:EK2} below)

 The following lemma is the heart of the analysis, 
  which is a version of Lemma 4.3 of \cite{OK23} for random generation (see also Appendix~\ref{sec:proof-lemm5.2}, for a proof). 

\begin{lemma}
  \label{lemm:average_space_complexity}
  Let $\mu \in (1, 2)$ be rational given by an irreducible fraction $\mu = c/d$.
  Suppose for $\mu \in (1, 2)$ that $\tent_{\mu}^{i}(\frac{1}{2}) \neq \frac{1}{2}$ holds for any $i = 1, \dots, n-1$. 
  On condition that $X$ fully covers $(x-\epsilon + \frac{1}{n^2}, x + \epsilon - \frac{1}{n^2})$ at $\lceil \log_{\mu} n \rceil$, 
  the conditional expectation of $K^2$ is 
    $\Order(\log_{\mu}^{2}{n} \log_{\mu}^{2}{d}) = \Order(\lg^2 n \lg ^2 d / \lg^4 \mu)$. 
\end{lemma}
 Lemma~\ref{lemm:average_space_complexity} is supported by the following Lemma~\ref{lem:trans-prob}, 
  which is almost trivial from the fact that the iterative tent map $\tent^i$ is piecewise linear 
 (see Appendix~\ref{apx:trans-prob} for a proof). 
\begin{lemma}\label{lem:trans-prob}
Let $X \in [x-\epsilon,x+\epsilon]$ uniformly at random. 
Let $B_1 \cdots B_n = \enc^n(X)$. 
Suppose $X$ fully covers  $y \in[x-\epsilon,x+\epsilon]$ at $i$, and let $\enc^n(y) = b_1\cdots b_n$. 
Then, 
 $\Pr[B_{i+1} = b_{i+1} \mid \enc^i(X)= \enc^i(y)] = \frac{|\typefn(b_1\cdots b_i b_{i+1})|}{\mu |\typefn(b_1\cdots b_i)|}$ holds. 
\end{lemma}

Lemma~\ref{lem:EK2} is easy from Lemma~\ref{lemm:average_space_complexity}, as follows. 
\begin{proof}[Proof of Lemma~\ref{lem:EK2}.]
 Note that  the probability of the event $X \not\in (x-\epsilon + \frac{1}{n^2}, x + \epsilon - \frac{1}{n^2})$ is at most $\frac{2}{n^2}$. 
 Using the trivial upper bound that $K \leq n$, 
  the claim is easy  from  
 Lemma~\ref{lemm:average_space_complexity}. 
\end{proof}

\section{Concluding Remarks}
 Motivated by the possibility of a {\em valid} computation of physical systems, 
   this paper investigated the space complexity of computing a tent code. 
 We showed that a valid approximate calculation is in $\Order(\log n)$ space, that is optimum in terms of $n$,  and 
   gave an algorithm for the valid decision working in $\Order(\log^2 n)$ space, 
   in a sense of the smoothed complexity where the initial condition $x'$ is $\epsilon$-perturbed from $x$. 
 A future work is an extension to the baker's map, 
   which is a chaotic map of piecewise but 2-dimensional. 
 For the purpose, we need an appropriately extended notion of the segment-type. 
 Another future work is an extension to the logistic map, 
  which is a chaotic map of 1-dimensional but quadratic. 
 The time complexity of the tent code is another interesting topic 
   to decide $b_n\in \{0,1\}$ as given a rational $x=p/q$ for a fixed $\mu \in \mathbb{Q}$. 
 Is it possible to compute in time polynomial in the input size $\log p+\log q+\log n$? 
 It might be NP-hard, but we could not find a result. 
 
\bibliographystyle{plain}

\appendix
\section{Subprocesses}\label{apx:subpro}
This section shows two subprocesses Algorithms~\ref{alg:vucdelt} and~\ref{alg:computel}, 
 which are called in Algorithms~\ref{alg:calc-improve} and~\ref{alg:recog}.
 Algorithm~\ref{alg:vucdelt} follows  Lemmas~\ref{lem:transition} and \ref{lem:typei}, and 
 Algorithm~\ref{alg:computel} follows \eqref{def:delta}.

\begin{algorithm}[h]
    \caption{Compute $v,u,c,\delta$}
    \label{alg:vucdelt}
    \begin{algorithmic}[1]
    \REQUIRE $k$   
    \ENSURE $v[k]$, $u[k]$, $c[k]$, $\delta[k-1,0]$, $\delta[k-1,1]$ 
    \IF{$k=-1$}
    \STATE $v[-1] \leftarrow 0$, $u[-1] \leftarrow 0$ \COMMENT{reject state}
    \ENDIF
    \IF{$k=0$}
    \STATE $v[0] \leftarrow 0$, $u[0] \leftarrow 1$, {\rm bit} $c[0] \leftarrow 0$ \COMMENT{$=q_0$}
    \ENDIF
     \IF{$k=1$}
    \STATE $v[1] \leftarrow 0$, $u[1] \leftarrow \tent(\frac{1}{2})$, $c[1] \leftarrow 1$, $\delta[0,0]\leftarrow 1$, $\delta[0,1]\leftarrow 1$   \COMMENT{$=\typei_1$}
    \ENDIF
    \IF{$k \geq 2$}
    \IF{$v[k-1] < \frac{1}{2} < u[k-1]$}
    \IF{$c[k-1]=0$}
    \STATE $\delta[k-1,0] \leftarrow k$, $v[k] \leftarrow \tent(v[k-1])$, $u[k] \leftarrow \tent(\tfrac{1}{2})$, $c[k] \leftarrow 0$
    \STATE $\delta[k-1,1] \leftarrow k'$ such that $v[k'] = \tent(u[k-1])$ and $u[k'] =\tent(\tfrac{1}{2})$
    \ELSE[i.e., $c\lbrack k-1 \rbrack=1$]
    \STATE $\delta[k-1,0] \leftarrow k$, $v[k] \leftarrow \tent(u[k-1])$, $u[k] \leftarrow \tent(\tfrac{1}{2})$, $c[k] \leftarrow 0$
    \STATE $\delta[k-1,1] \leftarrow k'$ such that $v[k'] = \tent(v[k-1])$ and $u[k'] =\tent(\tfrac{1}{2})$
    \ENDIF
    \ELSIF{$u[k-1] \leq \frac{1}{2}$}
    \STATE $\delta[k-1,c[k-1]] \leftarrow k$, $v[k] \leftarrow \tent(v[k-1])$, $u[k] \leftarrow \tent(u[k-1])$, $c[k] \leftarrow c[k-1]$
    \STATE $\delta[k-1,\flip{c[k-1]}] \leftarrow -1$
    \ELSE[i.e., $v \lbrack k-1 \rbrack \geq \frac{1}{2}$]
    \STATE $\delta[k-1,\flip{c[k-1]}] \leftarrow k$, $v[k] \leftarrow \tent(u[k-1])$, $u[k] \leftarrow \tent(v[k-1])$, $c[k] \leftarrow \flip{c[k-1]}$
    \STATE $\delta[k-1,c[k-1]] \leftarrow -1$
    \ENDIF
    \ENDIF
    \end{algorithmic}
\end{algorithm}
\begin{algorithm}[h]
    \caption{Update $l$ and $b$}
    \label{alg:computel}
    \begin{algorithmic}[1]
    \REQUIRE an integer $l$, bits $b$, $b'$  
    \ENSURE an integer $l$, a bit $b$ 
    \IF[$Z_i =\typei_l$]{$b=c[l]$}
    \STATE $l \leftarrow \delta[l,b']$, $b \leftarrow b'$
    \ELSE[$Z_i = \flip{\typei}_l$]
    \STATE $l \leftarrow \delta[l,\flip{b'}]$, $b \leftarrow \flip{b'}$
    \ENDIF
    \RETURN $l$ and $b$
    \end{algorithmic}
\end{algorithm}

\section{Proof of Lemma ~\ref{lemm:average_space_complexity}}\label{sec:proof-lemm5.2}
This Section proves Lemma~\ref{lemm:average_space_complexity}. 
\begin{lemma}[Lemma~\ref{lemm:average_space_complexity}]
  Let $\mu \in (1, 2)$ be rational given by an irreducible fraction $\mu = c/d$.
  Suppose for $\mu \in (1, 2)$ that $\tent_{\mu}^{i}(\frac{1}{2}) \neq \frac{1}{2}$ holds for any $i = 1, \dots, n-1$. 
  On condition that $X$ fully covers $(x-\epsilon + \frac{1}{n^2}, x + \epsilon - \frac{1}{n^2})$ at $\lceil \log_{\mu} n \rceil$, 
  the conditional expectation of $K^2$ is 
    $\Order(\log_{\mu}^{2}{n} \log_{\mu}^{2}{d}) = \Order(\lg^2 n \lg ^2 d / \lg^4 \mu)$. 
\end{lemma}

 The proof strategy of  Lemma~\ref{lemm:average_space_complexity} is as follows. 
 Lemma~\ref{lemm:transition_function}
   implies that a chain must follow the path $\typei_l, \typei_{l+1}, \ldots, \typei_{2l}$ 
     (or $\flip{\typei}_l, \flip{\typei}_{l+1}, \ldots, \flip{\typei}_{2l}$) to reach level $2l$ 
    and the probability is $\frac{|\typei_{2l}|}{\mu^{l}|\typei_{l}|}$ (Lemma~\ref{lem:go_back}). 
 We then prove that 
    there exists $l=\Order(\log n \log d)$ such that $\frac{|\typei_{2l}|}{\mu^{l}|\typei_{l}|} \leq n^{-3}$
      (Lemma~\ref{lemm:existence_of_short_q_s}), 
  which provides $\Pr[K \geq 2l] \leq n^{-2}$ (Lemma~\ref{lem:l*}).  
 Lemma~\ref{lemm:average_space_complexity} is easy from  Lemma~\ref{lem:l*}. 
   
 Let $Z_t = L(\typefn(\enc^t(X))$ for $t=0,1,2,\ldots$, i.e., $Z_t$ denote the level of the state at $t$-th iteration. 
 We observe the following fact from Lemma~\ref{lemm:transition_function}. 
\begin{observation}\label{obs:go_straight}
  If $Z_t$ visits $\typei_{2j}$ (resp.\ $\flip{\typei}_{2j}$) {\em for the first time} then 
    $Z_{t-i} = \typei_{2j-i}$ (resp.\ $Z_{t-i} = \flip{\typei}_{2j-i}$) for $i=1,2,\ldots,j$. 
\end{observation}
\begin{proof}
 By Lemma~\ref{lemm:transition_function}, 
  all in-edges to $\typei_k$ (resp.\ $\flip{\typei}_k$) for any $k = j+1,\dots,2j$ 
  come from $\typei_{k-1}$ (resp.\ $\flip{\typei}_{k-1}$),  or a node of level $2j$ or greater. 
 Since $Z_t$ has not visited any level greater than $2j$ by the hypothesis and the above argument again, 
 we obtain the claim. 
\end{proof}

 By Observation~\ref{obs:go_straight}, 
   if a Markov chain $Z_1,Z_2,\ldots$ visits level $2l$ {\em for the first time} at time $t$ 
     then $\lev(Z_{t-l})$ must be $l$. 
 The next lemma gives an upper bound of the probability from level $l$ to $2l$. 
\begin{lemma}\label{lem:go_back}
Suppose that $X$ covers around appropriate $y$ corresponding to $Z_{t-l}$ at $t-l$. Then, 
$$\Pr[\lev(Z_{t}) = 2l \mid \lev(Z_{t-l}) = l] 
 =  \frac{|\typei_{2l}|}{\mu^{l} |\typei_{l}|}.$$
\end{lemma}
\begin{proof}
 By Observation~\ref{obs:go_straight}, 
  the path from $\typei_l$ to  $\typei_{2l}$ is unique and 
  \begin{align}
    \label{eq:asc_1_7}
 \Pr[Z_{t}  = \typei_{2l} \mid Z_{t-l} = \typei_l] 
    &= \prod_{i=l}^{2l-1} p(\typei_{i}, \typei_{i+1}) \nonumber \\
    &= \prod_{i = l}^{2l-1} \frac{|\typei_{i+1}|}{\mu |\typei_{i}|} && (\mbox{by Lemma \ref{lem:trans-prob}}) \nonumber \\
    &= \frac{|\typei_{2l}|}{\mu^{l} |\typei_{l}|}
  \end{align}
holds. 
 We remark that $|\typei_i| = |\flip{\typei}_i|$ holds for any $i$, meaning that 
  $p(\typei_{i}, \typei_{i+1}) = p(\flip{\typei}_{i}, \flip{\typei}_{i+1})$, 
   and hence 
  $\Pr[Z_{t}  =  \flip{\typei}_{2l} \mid Z_{t-l} = \flip{\typei}_l] = \frac{|\typei_{2l}|}{\mu^{l} |\typei_{l}|}$. 
\end{proof}

 The following lemma is the first mission of the proof of Lemma~\ref{lemm:average_space_complexity}. 
\begin{lemma}
  \label{lemm:existence_of_short_q_s}
  Let $\mu \in (1, 2)$ be rational given by an irreducible fraction $\mu = c/d$.
 Suppose for $\mu \in (1, 2)$ that $\tent_{\mu}^{i}(\frac{1}{2}) \neq \frac{1}{2}$ holds for any $i = 1, \dots, n-1$. 
 Then, there exists $l$ such that  $l \leq 8 \ceil{\log_{\mu} d} \ceil{ \log_{\mu} n} $ and 
  \begin{equation}
    \label{eq:shortqs_0}
    \frac{|\typei_{2l}|}{\mu^l |\typei_l|} \le n^{-3}
  \end{equation}
  holds. 
\end{lemma}
To prove Lemma \ref{lemm:existence_of_short_q_s}, 
 we remark the following fact. 
\begin{lemma}
  \label{lemm:min_sectype_length}
  Let $\mu \in (1, 2)$ be rational given by an irreducible fraction $\mu = c/d$.
  Then, $|\typei_k| \ge \frac{1}{2 d^k}$ for any $k \ge 2$. 
\end{lemma}

\begin{proof}
 By the recursive formula \eqref{eq:20231105a}, 
  we see that $\typei_k$ is either 
   $\left[\tent^{i}(\frac{1}{2}), \tent^{j}(\frac{1}{2})\right)$ or $\left(\tent^{i}(\frac{1}{2}), \tent^{j}(\frac{1}{2})\right]$ where $i \le k$ and $j \le k$.
  We can denote $\tent^{i}(\frac{1}{2})$ as $\frac{c_{i}}{2d^{i}}$ ($c_{i} \in \mathbb{Z}_{>0}$) for any $i$.
  Therefore,
  \begin{equation}
    |\typei_k|
    = \left|\tent^{i}(\tfrac{1}{2}) - \tent^{j}(\tfrac{1}{2})\right|
    = \left|\frac{c_{i}}{2 d^{i}} - \frac{c_{j}}{2 d^{j}}\right|
    = \left|\frac{c_{i} d^{k-i} - c_{j} d^{k-j}}{2 d^k}\right|
  \end{equation}
  holds.
  Clearly, $c_{i} d^{k-i} - c_{j} d^{k-j}$ is an integer, and it is not $0$ since $|\typei_k| \neq 0$.
  Thus, we obtain $|\typei_k| \ge \frac{1}{2 d^k}$.
\end{proof}

 Then, we prove Lemma~\ref{lemm:existence_of_short_q_s}. 
\begin{proof}[Proof of Lemma~\ref{lemm:existence_of_short_q_s}]
  For convenience, let 
   $l_i =  2^i \ceil{\log_{\mu}n} $ for $i = 1, 2, \ldots$. 
  Assume for a contradiction that 
     \eqref{eq:shortqs_0} never hold for any $l_1,l_2,\ldots,l_k$, where 
  $k = \max\{4, \ceil{\log_{2} \log_{\mu} d} + 2\}$ for convenience.  
 In other words, 
  \begin{equation}
    |\typei_{l_{i+1}}| > n^{-3} \mu^{l_{i}} |\typei_{l_{i}}|
  \end{equation}
  holds every $i=1,2,\ldots,k$. 
 Thus, we inductively obtain that 
  \begin{align}
    |\typei_{l_{k+1}}|
    &> n^{-3} \mu^{l_{k}} |\typei_{l_{k}}| \nonumber \\
    &> n^{-6} \mu^{l_{k}} \mu^{l_{k-1}} |\typei_{l_{k-1}}| \nonumber \\
    &> \dots \nonumber \\
    &> n^{-3k} \mu^{l_{k}} \mu^{l_{k-1}} \dots \mu^{l_1} |\typei_{l_1}|
    \label{eq:shortqs_2}
  \end{align}
   holds. 
 By the definition of $l_i$, 
  \begin{equation}
    \label{eq:shortqs_5}
    \mu^{l_{i}} = \mu^{2^{i} \ceil{\log_{\mu}n}} \ge \mu^{2^{i} \log_{\mu}n} = n^{2^{i}}
  \end{equation}
  holds. 
 Lemma \ref{lemm:min_sectype_length} implies that 
  \begin{align}
    \label{eq:shortqs_7}
    |\typei_{l_{1}}|
    \ge \frac{1}{2 d^{l_{1}}}  
    = \frac{1}{2} d^{-2 \ceil{\log_{\mu} n}} 
    \geq \frac{1}{2} d^{-4 \log_{\mu} n} 
    = \frac{1}{2} n^{-4 \log_{\mu} d}
  \end{align}
 holds.  
 Then,  \eqref{eq:shortqs_2},  \eqref{eq:shortqs_5} and \eqref{eq:shortqs_7} imply 
  \begin{align}
    |\typei_{l_{k+1}}|
    &>  n^{-3k}  \cdotp  n^{2^{k} + 2^{k-1} + \dots + 2^{1}} \cdotp \frac{1}{2} n^{-4 \log_{\mu} d} 
    \label{eq:shortqs_6}
  \end{align}
 holds. 
  By taking the $\log_{n}$ of the both sides of \eqref{eq:shortqs_6}, 
    we see that 
  \begin{align}
    \log_{n}{|\typei_{l_{k+1}}|}
    &> -3k + 2^{k+1} - 2 - \log_{n}2 - 4 \log_{\mu}d  \nonumber \\
    &= (2^{k} - 4 \log_{\mu}d) + (2^{k} -3k - 2 - \log_{n}2)
    \label{eq:shortqs_8}
  \end{align}
   holds. 
 Since  $k \geq  \ceil{\log_{2} \log_{\mu} d} + 2$ by definition,  
  it is not difficult to see  that 
  \begin{align}
    2^{k} - 4 \log_{\mu}d 
    &\geq 2^{2 + \log_{2}(\log_{\mu}d)} - 4 \log_{\mu}d \nonumber \\
    &= 4 \log_{\mu}d - 4 \log_{\mu}d \nonumber \\
    &= 0
    \label{eq:shortqs_1}
  \end{align}
  holds.
 Since  $k \geq  4$ by definition,  
 it is also not difficult to observe that 
  \begin{align}
    2^{k} -3k - 2 - \log_{n} 2
   \ \ge\ 2^{4} - 3 \cdot 4 - 2 - \log_{n} 2 
   \ =\ 2- \log_{n} 2
   \ >\ 0
    \label{eq:shortqs_9}
  \end{align}
  holds.  
  Equations \eqref{eq:shortqs_8}, \eqref{eq:shortqs_1} and \eqref{eq:shortqs_9} imply that $\log_{n}{|\typei_{l_{k+1}}|} > 0$, 
     meaning that $|\typei_{l_{k+1}}| > 1$. 
  At the same time, 
    notice that any segment-type $\typej$ satisfies $\typej \subseteq [0,1]$,  
     meaning that 
    $|\typei_{l_{k+1}}| \le 1$. 
  Contradiction. 
 Thus, we obtain \eqref{eq:shortqs_0} for at least one of $l_1,l_2,\ldots,l_k$.  

 Finally, we check the size of $l_k$:  
\begin{align*}
  l_k &= 2^k \ceil{\log_{\mu} n} \leq 2^{\max\{4,\ceil{\log_2 \log_{\mu} d} + 2\}} \ceil{\log_{\mu} n} 
  \leq 2^{\max\{4, \log_2 \log_{\mu} d + 3\}} \ceil{\log_{\mu} n} \\& = \max\{ 16, 8 \log_{\mu} d \} \ceil{\log_{\mu} n}  = 8 \max\{ 2, \log_{\mu} d \} \ceil{\log_{\mu} n} \leq 8 \ceil{\log_{\mu} d} \ceil{\log_{\mu} n} 
\end{align*}
 where the last equality follows $\log_{\mu} d > 1$ since $\mu < 2$ and $d \geq 2$. 
 We obtain a desired $l$. 
\end{proof}

By Lemmas~\ref{lem:go_back} and \ref{lemm:existence_of_short_q_s}, we obtain the following fact. 
\begin{lemma}\label{lem:l*}
Let $l_* = 8 \ceil{\log_{\mu} d} \ceil{\log_{\mu} n} $ for convenience. Then 
\begin{align*}
 \Pr[\requiredstgsize \ge 2 l_*]  \leq n^{-2} 
\end{align*} 
 holds. 
\end{lemma}
\begin{proof}
 For $\mu = c/d$ and $n$, Lemma~\ref{lemm:existence_of_short_q_s} implies that there exists $l$ such that 
  $l \leq l_* $ and 
\begin{align}
 \frac{|\typei_{2l}|}{\mu^{l} |\typei_{l}|} \leq n^{-3}
\label{eq:l*1}
\end{align}
holds.  
 Let $A_t$ ($t=1,\ldots,n$) denote the event that $Z_t$ reaches the level $2 l$ {\em for the first time}. 
 It is easy to see that 
\begin{align}
 \Pr[\requiredstgsize \ge 2l] 
  = \Pr\left[\bigvee_{t=0}^{n} A_t \right] 
\label{eq:l*2}
\end{align}
 holds\footnote{Precisely, $\bigvee_{t=0}^{n} A_t = \bigvee_{t=2l_*}^{n} A_t$ holds, but we do not use the fact here.  } 
   by the definition of $A_t$. 
 We also remark that the event $A_t$ implies not only $Z_t = 2l$ but also $\lev(Z_{t-l}) = l$ by Observation~\ref{obs:go_straight}.  
 It means that 
\begin{align}
 \Pr[A_t] 
  &\leq \Pr[ [\lev(Z_t) = 2l] \wedge [\lev(Z_{t-l}) = l] ] 
  \nonumber\\
  &= \Pr[ \lev(Z_t) = 2l  \mid \lev(Z_{t-l}) = l ] \Pr[ \lev(Z_{t-l}) = l ] \nonumber\\
  &\leq \Pr[ \lev(Z_t) = 2l  \mid \lev(Z_{t-l}) = l ] 
\label{eq:l*3}
\end{align}  
 holds. Then, 
\begin{align*}
 \Pr[\requiredstgsize \ge 2l] 
 &= \Pr\left[\bigvee_{t=0}^{n} A_t \right] 
   && (\mbox{by \eqref{eq:l*2}}) \nonumber \\
 &\le \sum_{t=0}^n \Pr\left[A_t\right] 
   && (\mbox{union bound})  \nonumber \\
  &\leq n \Pr[ \lev(Z_t) = 2l  \mid \lev(Z_{t-l}) = l ] 
   && (\mbox{by \eqref{eq:l*3}}) \nonumber \\
  &\leq n \frac{|\typei_{2l}|}{\mu^{l} |\typei_{l}|}
   && (\mbox{by Lemma~\ref{lem:go_back}}) \nonumber \\
  & \leq n^{-2}
   && (\mbox{by \eqref{eq:l*1}}) 
\end{align*}
  holds. 
We remark that 
$\Pr[\requiredstgsize \ge 2l_*] 
 \leq 
 \Pr[\requiredstgsize \ge 2l] $
is trivial since $l < l_*$. 
\end{proof}

We are ready to prove Lemma \ref{lemm:average_space_complexity}. 
\begin{proof}[Proof of Lemma \ref{lemm:average_space_complexity}]
Let $l_* = 8 \ceil{\log_{\mu} d} \ceil{\log_{\mu} n} $ for convenience. 
Notice that $X$ fully covers $(x-\epsilon + \frac{1}{n^2}, x + \epsilon - \frac{1}{n^2})$
 at or after $l_* \geq \lceil 2\log_{\mu} n \rceil$ by Lemma~\ref{lem:cover}. 
Then 
\begin{align*}
  \E[\requiredstgsize^{2}]
   &= \sum_{k=1}^n k^2\Pr[\requiredstgsize =k]  \nonumber \\
   &= \sum_{k=1}^{2l_*-1} k^2\Pr[\requiredstgsize =k]  + \sum_{k=2l_*}^n k^2\Pr[\requiredstgsize =k]  \nonumber \\
   &\leq (2l_* - 1)^2 \Pr[\requiredstgsize \leq 2l_*-1] + n^{2} \Pr[\requiredstgsize \ge 2l_*] \nonumber \\
   &\leq (2l_* - 1)^2 + n^{2} \Pr[\requiredstgsize \ge 2l_*] \nonumber \\
   &\leq (2l_* - 1)^2 + 1
    && (\mbox{by Lemma~\ref{lem:l*}}) \\
   &= (16 \ceil{\log_{\mu} d} \ceil{\log_{\mu} n} -1)^2 + 1
  \end{align*}
holds. Now the claim is easy. 
\end{proof}

\section{Proof of Lemma~\ref{lem:trans-prob}}\label{apx:trans-prob}
\begin{lemma}[Lemma~\ref{lem:trans-prob}]
Let $X \in [x-\epsilon,x+\epsilon]$ uniformly at random. 
Let $B_1 \cdots B_n = \enc^n(X)$. 
Suppose $X$ covers around $y \in[x-\epsilon,x+\epsilon]$ at $i$, and let $\enc^n(y) = b_1\cdots b_n$. 
Then, 
\begin{align*}
 \Pr[B_{i+1} = b_{i+1} \mid \enc^i(X)= \enc^i(y)] = \frac{|\typefn(b_1\cdots b_i b_{i+1})|}{\mu |\typefn(b_1\cdots b_i)|}
\end{align*} 
  holds. 
\end{lemma}
\begin{proof}
Since $X$ covers around $y$ at $i$, it is not difficult to see that 
\begin{align*}
&\Pr[ \enc^i(X) = \bitseq_i]
       = \frac{|S_i|}{2\epsilon} \quad \mbox{and} \\
&\Pr[ \enc^{i+1}(X) = \bitseq_{i+1}, \enc^i(X) = \bitseq_i]
= \Pr[ \enc^{i+1}(X) = \bitseq_{i+1}]
       = \frac{|S_{i+1}|}{2\epsilon}
\end{align*}
 hold. Thus, 
\begin{align}
\Pr[ \enc^{i+1}(X) = \bitseq_{i+1} \mid \enc^i(X) = \bitseq_i]
       = \frac{|S_{i+1}|}{|S_i|}
       = \frac{\mu^{i+1} |S_{i+1}|}{\mu^{i+1} |S_i|}
\label{eq:20231111a}
\end{align}
   holds. 
 Since the iterative tent map $\tent^i$ is piecewise linear, 
  it is not difficult to see that 
   $\tent^i$ preserves the uniform measure (cf.\ Lemma B.6 in \cite{OK23}), and hence 
  $\mu^i |S_i| = |\typefn(\enc^i(y))|$ as well as 
   $\mu^{i+1} |S_{i+1}| = |\typefn(\enc^{i+1}(y))|$ hold. 
 Then, 
\begin{align*}
\eqref {eq:20231111a} = \frac{|\typefn(\enc^{i+1}(y))|}{\mu|\typefn(\enc^i(y))|}
\end{align*}
 and we obtain the claim. 
\end{proof}
\end{document}